\documentclass[11pt, reqno, a4paper]{amsart}
\usepackage{microtype}
\usepackage[scr=euler]{mathalfa}
\usepackage[british]{babel}
\usepackage{textcomp}
\usepackage{tikz-cd}
\usepackage{amsmath, amssymb, amsthm}
\usepackage{mathtools}
\usepackage{bussproofs}
\usepackage[hang,flushmargin]{footmisc}

\usepackage[a4paper,margin=3.1cm]{geometry}

\newcommand{\eg}{e.g.~}
\newcommand{\ie}{i.e.~}
\usepackage{colonequals}
\newcommand{\defeq}{\colonequals}

\usepackage[hyphens]{url}
\usepackage[breaklinks=true]{hyperref}
\usepackage{cleveref}

 \usetikzlibrary{arrows,calc,patterns,positioning,shapes}
   \usetikzlibrary{decorations.pathmorphing}
   \tikzset{
       modal/.style={shorten >=1pt,shorten <=1pt,auto,
                     node distance=1.5cm,semithick},
       world/.style={circle,draw,minimum size=1cm,fill=gray!15},
       point/.style={circle,draw,fill=black,inner sep=0.5mm},
       reflexive/.style={->,in=120,out=60,loop,looseness=#1},
       reflexive/.default={5},
       reflexive point/.style={->,in=135,out=45,loop,looseness=#1},
       reflexive point/.default={25},
       }

\usepackage[framemethod=TikZ]{mdframed}
\mdfsetup{%
middlelinecolor=red,
middlelinewidth=2pt,
backgroundcolor=red!10,
roundcorner=10pt}

\usepackage{ccc_macros}

\begin{document}

\title[Combining contextuality and causality]{Combining contextuality and causality:\\ a game semantics approach}

\author[S.~Abramsky]{Samson Abramsky}
\address{Samson Abramsky \newline\indent Department of Computer Science, University College London \newline\indent 66--72 Gower Street, London WC1E 6EA, United Kingdom}
\email{s.abramsky@ucl.ac.uk}
\urladdr{http://www.cs.ucl.ac.uk/people/S.Abramsky/}
\author[R.S.~Barbosa]{Rui Soares Barbosa}
\address{Rui Soares Barbosa \newline\indent INL -- International Iberian Nanotechnology Laboraory \newline\indent Av. Mestre Jos\'e Veiga, 4715-330 Braga, Portugal}
\email{rui.soaresbarbosa@inl.int}
\urladdr{https://www.ruisoaresbarbosa.com/}
\author[A.~Searle]{Amy Searle}
\address{Amy Searle \newline\indent  Department of Physics, University of Oxford \newline\indent Clarendon Laboratory, Parks Road, Oxford OX1 3PU, United Kingdom}
\email{amy.searle@physics.ox.ac.uk}
\urladdr{https://www.physics.ox.ac.uk/our-people/searle}

\begin{abstract}
We develop an approach to combining contextuality with causality, which is general enough to cover causal background structure, adaptive measurement-based quantum computation, and causal networks. 
The key idea is to view contextuality as arising from a game played between Experimenter and Nature, allowing for causal dependencies in the actions of both the Experimenter (choice of measurements) and Nature (choice of outcomes).
\end{abstract}

\maketitle

\section{Introduction}
Contextuality is a key non-classical feature of quantum theory.
Besides its importance in quantum foundations, it has been linked to quantum advantage in information-processing tasks.
It also arises beyond quantum mechanics, cf. \cite{dzhafarov2015contextuality}.

We wish to generalise contextuality to accommodate \emph{causality} and \emph{adaptivity}.
These features may arise from:
\begin{itemize}
\item fundamental aspects of the physical setting, in particular the causal structure of spacetime;
\item the causal structure of an experiment, where measurements are performed in some causal order, and moreover, which measurements are performed may depend on the outcomes of previous measurements;
\item feed forward in measurement-based quantum computation (MBQC) \cite{briegel2009measurement}, and more generally, adaptive computation.
\end{itemize}
Our objectives include:
\begin{itemize}
\item A more fine-grained analysis of contextuality. 
Signalling should be allowed from the causal past, \ie the backward light cone, and thus no-signalling/no-disturbance should be imposed only from \emph{outside} it.
This in turn modifies the scope of classicality (non-contextuality), which now becomes relative to this weaker form of no-signalling constraints.
\item A better connection with computational models such as circuits and MBQC. Explicitly representing causal flows of information, \eg outputs of gates feeding into inputs of other gates, enables a deeper analysis of the relationships between contextuality and quantum advantage.
\end{itemize}

It turns out that capturing these different manifestations of causality and their interactions with contextuality is rather subtle.
The perspective we adopt here is to view contextuality as a two-person game played between Experimenter and Nature.
The Experimenter's moves are the measurements; \ie the actions of the Experimenter are to choose the next measurement to be performed. Nature's moves are the outcomes.
We can capture the various forms of causal dependency which may arise in terms of \emph{strategies} for Experimenter or for Nature. 

The game format is already familiar in the form of non-local games. 
There, the Verifier plays the role of the Experimenter, and Nature responds with outcomes according to the probability distributions 
corresponding to Alice--Bob strategies. 
Non-local games are one-shot games, with a single round of interaction. By considering more general games, causal structure can be incorporated.

Our treatment builds upon the sheaf-theoretic approach to contextuality. A pleasing feature is that once one modifies the basic sheaf of events to take causal structure into account, the further definitions and treatment of contextuality follow \textit{automatically}.
This illustrates the advantages of a compositional and functorial approach.

\section{Previous work}

Pearl had already noted the connection with Bell inequalities in his seminal paper on testability of causal models with latent and instrumental variables \cite{DBLP:conf/uai/Pearl95}.
The extension of causal networks to allow for quantum resources, or more generally the operations offered by Generalised Probabilistic Theories, has been studied \eg in \cite{henson2014theory,chaves2018quantum}. 

Our starting point is the sheaf-theoretic treatment of contextuality introduced in \cite{abramsky2011sheaf}, and extensively developed subsequently.
This is a general, mathematically robust approach, which provides a basis for:
\begin{itemize}
\item the contextual fraction as a measure of contextuality \cite{abramsky2017contextual};
\item a general characterisation of noncontextuality inequalities in terms of consistency conditions (``logical Bell inequalities'', Boole's ``conditions of possible experience'') \cite{abramsky2012logical,abramsky2020classical};
\item resource theory of contextuality, and simulations between contextual systems \cite{abramsky2017contextual,karvonen2018categories, abramsky2019comonadic,karvonen2021neither,barbosa2021closing};
\item cohomological criteria for contextuality, the topology of contextuality \cite{abramsky2012cohomology,abramsky2015ccp,caru2017cohomology};
\item connections with logic and computation, database theory, constraint satisfaction \cite{abramsky2015contextual,conghaile2022cohomology};
\item generalisations \cite{barbosa2015DPhil,atserias2020consistency} and applications \cite{barbosa2014monogamy} of Vorob{\textquotesingle}ev's theorem \cite{vorobev1962consistent}.
\end{itemize}
The aim is to develop a refined version incorporating causality for which all these features will carry over.

There have been some prior works in this direction:
\begin{itemize}
\item Shane Mansfield in \cite{SM2017} introduced a refinement of the sheaf-theoretic approach with an order on the measurements,
and used it to study the two-slit experiment and the Leggett--Garg scenario.

\item Stefano Gogioso and Nicola Pinzani  in \cite{gogioso2021sheaf,gogioso2023geometry} developed a causal refinement of the sheaf-theoretic approach to non-locality, \ie for the case of  \emph{Bell-type scenarios}. 
They introduce an order on the sites or agents in the Bell scenario.
\end{itemize}
In both cases, the order is used to refine the no-signalling or no-disturbance condition which guarantees that joint distributions have consistent marginals.
In the presence of causality, signalling is allowed from within the backwards light cone or causal past of an event, and thus no-signalling is only required outside it.

One may contrast this with  the Contextuality-by-Default (CbD) approach introduced by Ehtibar Dzhafarov and Janne Kujala  \cite{dzhafarov2016context,dzhafarov2016contextuality}.
In CbD, \emph{every} variable is regarded as contextual, differently labelled in each context. 
Classicality is characterised by the existence of a joint distribution under which different occurrences of variables with the same ``content'' have the same value with the maximum probability consistent with their individual marginals.
This allows for the analysis of arbitrary signalling systems, which has applications e.g.~in the behavioural sciences, where signalling is the norm. Moreover, this signalling may in general be  impossible to characterise or control.

By contrast, both in the above work by Mansfield and Gogioso--Pinzani and in the present paper, the aim is to \emph{explicitly describe} a given causal background -- which might arise from the structure of an experiment, circuit, or physical system -- and to characterise contextuality relative to such a background.

In this paper, we extend the scope of previous work in several directions.
First, we  allow more general dependencies of events on their prior causal histories.
In particular, the choice of which measurement to perform can depend on previous outcomes as well as on which measurements have been performed. This is an important feature of MBQC (``feedforward''), and more generally of adaptive computation.
Secondly, we extend general contextuality scenarios with causality, not just the non-locality Bell scenarios as in the Gogioso--Pinzani (GP) approach.
Finally, and most subtly, we recognise the different roles played by Nature and Experimenter in their causal interactions,  highlighting an important difference between causal background and adaptivity.

An interesting feature of our approach, in common with that of Gogioso--Pinzani, is that it proceeds essentially by modifying the sheaf of events from \cite{abramsky2011sheaf} to reflect the refined signalling constraints in the presence of causality.
Once this has been done, the remainder of the analysis of contextuality follows exactly the same script as in \cite{abramsky2011sheaf}. 
In particular, the appropriate definition of empirical model, the relaxed no-signalling constraints, and the notion of classicality/non-contextuality follow automatically.

\section{Examples}

As we have already suggested, causality in relation to contextuality has dual aspects. It may be imposed by Nature, in the form of a causal background against which the contextual behaviour plays out; or it may be imposed by the Experimenter, \eg to achieve computational effects (adaptive computation).
We illustrate these two sources of causality in two basic examples.

\subsection{Example I: causal background \`a la GP}
Consider a standard bipartite nonlocality scenario,  \eg the  Bell--CHSH scenario: 
 two experimenters, Alice and Bob, with sets of local measurements $I_A$ and $I_B$, and outcome sets $O_A$ and $O_B$.
 We may think of these as `ìnputs'' and ``outputs''.

 We now introduce a variation, in which
we assume that Alice's events \emph{causally precede} those of Bob.
Thus Bob's backward light cone includes the events where Alice chooses a measurement and observes an outcome.

Whereas in a standard, causally ``flat'' scenario, we would have deterministic outcomes given by functions 
\[ s_A : I_A \to O_A, \quad s_B : I_B \to O_B, \]
with these causal constraints, we have functions 
\[ s_A : I_A \to O_A, \quad s_B : I_A \times I_B \to O_B .\]
That is, the responses by Nature to Bob's measurement may depend on the previous measurement made by Alice.\footnote{Note that, in a deterministic model, Nature ``knows'' what response it would have given for Alice's measurement, so there is no real dependency on this outcome.}

If we have measurements $x_1, x_2 \in I_A$, $y \in I_B$, then $\{ (x_1,0), (y,0) \}$ and $\{ (x_2,0), (y,1) \} $ are valid histories in a single deterministic model.
If we now go to distributions over such histories, say $d_{\{x,y\}}$ as a distribution over outcomes for the Alice measurement $x$ and the Bob measurement $y$, then
of the usual no-signalling/compatibility equations
\begin{align}
    d_{\{x,y\}} |_{\{x\}}  & =  d_{\{x\}} \label{firstns}
    \\
    d_{\{x,y\}} |_{\{y\}}  & =  d_{\{y\}} 
\end{align}
only \eqref{firstns} remains. In fact, $d_{\{y\}}$ is not even defined, since $\{ y\}$ is not a ``causally secured'' context: the measurement $y$ can never occur on its own without a preceding Alice measurement.

Thus no-signalling is relaxed in a controlled fashion.

\subsection{Example II: Anders--Browne}
\label{ABsec}
The Anders--Browne construction  \cite{anders2009computational} shows how we can use a form of Experimenter-imposed causality to promote two sub-universal computational models (Pauli measurements and mod-2 linear classical processing) to universal MBQC.

It uses the GHZ state as a resource state:
\[\mbox{GHZ} \; = \; \frac{\mid \uparrow \uparrow \uparrow  \rangle \; + \; \mid \downarrow \downarrow \downarrow \rangle}{\sqrt{2}} . \]
Performing local Pauli $X$ and $Y$ measurements, we obtain the following table of possible joint outcomes\footnote{The table shows only the possibilistic information, \ie the supports of the probability distributions on joint outcomes, which are uniform on each row.}
\begin{center}
\begin{tabular}{c|cccccccc}
&  $+++$ & $++-$ & $+-+$ & $+--$  & $-++$ & $-+-$ & $--+$ & $---$  \\ \hline
$X\,X\,X$ & $1$ & $0$ & $0$ & $1$ & $0$ & $1$ & $1$ & $0$  \\
$X\,Y\,Y$ & $0$ & $1$ & $1$ & $0$ & $1$ & $0$ & $0$ & $1$  \\
$Y\,X\,Y$ & $0$ & $1$ & $1$ & $0$ & $1$ & $0$ & $0$ & $1$  \\
$Y\,Y\,X$ & $0$ & $1$ & $1$ & $0$ & $1$ & $0$ & $0$ & $1$  \\
\end{tabular}
\end{center}
In terms of parities, \ie products of $\pm1$ outputs under the correspondence given by the group isomorphism $\langle\{0,1\},\oplus\rangle \cong \langle\{+1,-1\}, \cdot\rangle$, the support satisfies the following equations:
\[ \begin{array}{lllcr}
X_1 & X_2 & X_3 & = & +1 \hphantom{.}\\
X_1 & Y_2  &Y_3 & = & -1 \hphantom{.}\\
Y_1 & X_2 & Y_3 & = & -1 \hphantom{.}\\
Y_1 & Y_2 & X_3 & = & -1 .
\end{array}
\]
The idea is to use an Experimenter causal flow to implement \textsc{OR}.\footnotemark
Taking $X$ as $0$, $Y$ as $1$, we consider the measurements for Alice and Bob as inputs to an \textsc{OR} gate.
We then use the following simple 
mod-2 linear mapping (\textsc{XOR} on the bit representations) from the Alice--Bob measurements to determine Charlie's measurement:
\[ \begin{array}{lcr}
\mathrlap{0}\hphantom{0}, \mathrlap{0}\hphantom{0} & \mapsto & 0  \\
\mathrlap{0}\hphantom{0}, \mathrlap{1}\hphantom{0} & \mapsto & 1 \\
\mathrlap{1}\hphantom{0}, \mathrlap{0}\hphantom{0} & \mapsto & 1  \\
\mathrlap{1}\hphantom{0}, \mathrlap{1}\hphantom{0} & \mapsto & 0  \\
\end{array} \qquad \qquad
\begin{array}{lcr}
\mathrlap{X}\hphantom{X}, \mathrlap{X}\hphantom{X} & \mapsto & X \hphantom{.}\\
\mathrlap{X}\hphantom{Y}, \mathrlap{Y}\hphantom{Y} & \mapsto & Y \hphantom{.}\\
\mathrlap{Y}\hphantom{X}, \mathrlap{X}\hphantom{Y} & \mapsto & Y \hphantom{.}\\
\mathrlap{Y}\hphantom{Y}, \mathrlap{Y}\hphantom{X} & \mapsto & X .
\end{array}
\]
The output of the \textsc{OR} function is read off from the \textsc{XOR} of the three outcome bits.

\footnotetext{Anders-Browne in fact implemented NAND, using a different version of the GHZ state. Implementing any non-linear Boolean function in this setting suffices to achieve universality.}

We draw attention to the following two remarks.
\begin{itemize}
\item
This example illustrates causality that is purely employed by the Experimenter.
From Nature's point of view, it is just the standard (``causally flat'') GHZ construction.
\item The above describes a simplified ``one-shot'' implementation of a single \textsc{OR} gate.
To represent general logical circuits with embedded \textsc{OR} gates, using this construction as a building block,
really requires (classically computed) feedforward of measurement settings.
This means that there is full adaptivity at work, \ie dependence of measurement choices on prior measurement outcomes.
\end{itemize}

\section{Game semantics of causality}

We conceptualise the dual nature of causality as a two-person game, played between Experimenter and Nature:
\begin{itemize}
\item Experimenter’s moves are measurements to be performed;
\item Nature’s moves are the outcomes.
\end{itemize}
By formalising this, we develop a theory of causal contextuality that recovers: 
\begin{itemize}
\item the usual theory of contextuality in the ``flat'' case,
\item the Gogioso--Pinzani theory of non-locality in a causal background,
\item MBQC with adaptive computation,
\item classical causal networks,
\end{itemize}
as special cases, and more.

\subsection{Measurement scenarios}
We begin by briefly reviewing some basic ingredients of the sheaf-theoretic formulation of contextuality. For further details, see e.g.~\cite{abramsky2011sheaf}.

A (flat) measurement scenario is a pair $(X, O)$, where:
\begin{itemize}
\item $X$ is a set of \emph{measurements}.
\item $O = \{ O_x \}_{x \in X}$ is the set of possible \emph{outcomes} for each measurement.
\end{itemize}

An \emph{event} has the form $(x,o)$, where $x \in X$ and $o \in O_x$. It corresponds to the measurement $x$ being performed, with outcome $o$ being observed. 

Given a set of events $s$,  its domain is the set of measurements performed:
\[ \dom(s) \defeq \pi_1 s = \{ x \mid \exists o. \, (x,o) \in s \} . \]
We say that $s$ is \emph{consistent} if $(x,y), (x, y') \in s$ implies $y = y'$.
In this case, $s$ defines a function from the measurements in its domain to outcomes.

A consistent set of events is a \emph{section}.
We define the \emph{event sheaf} $\ES$ over sets of measurements: for each set $U \subseteq X$ of measurements, $\ES(U)$ is the set of sections whose domain is $U$; when $U \subseteq V$, there is a restriction map $\ES(V) \to \ES(U)$.
The functoriality of these restriction maps formalises the no-disturbance condition, or ``generalised no-signalling'', at the level of deterministic models. Generalised no-signalling of probabilistic (or possibilistic) models will then follow automatically when we compose with the appropriate distribution monad, cf.~\cite{abramsky2011sheaf}.

The sheaf property of the event sheaf -- that compatible families of local sections glue together to yield unique global sections -- corresponds to the fact that \emph{deterministic models are non-contextual}.\footnote{Note that if we drop no-signalling, as in the CbD approach, this no longer holds.}
When we pass to distributions over the event sheaf, 
the sheaf property no longer holds, and this is exactly how contextuality arises. More precisely, we extend the measurement scenario to a \emph{contextuality scenario} by specifying a cover of $X$; a failure of the sheaf property with respect to this cover constitutes a witness to contextuality.

Our general strategy to accommodate causality is to modify the definition of the event sheaf. After this, we essentially follow the same script as above to give an account of contextuality in the causal setting. A similar procedure is followed in \cite{gogioso2021sheaf,gogioso2023geometry}.

\subsection{Causal measurement scenarios}

A causal measurement scenario is a tuple $M=(X, O, {\vdash})$, where the additional ingredient is an \emph{enabling relation}
that expresses causal constraints.
The intended interpretation of $s \vdash x$, where $s \in \bigcup_{U \subseteq X}\ES(U)$ is a consistent set of events and $x \in X$ a measurement,
is that it is possible to perform $x$ after the events in $s$ have occurred.
Note that this constraint refers to the measurement outcomes as well as the measurements that have been performed.
This allows adaptive behaviours to be described.

Given such a causal measurement scenario $M$, we use it to generate a set of \emph{histories}. A history is a set of events that can happen in a causally consistent fashion. We associate each measurement $x$ with a unique event occurrence, so histories are required to be consistent.

To formalise this, we first define the \emph{accessibility relation} $\acc$ between consistent sets of events $s$ and measurements $x$: $s \acc x$ if and only if $x \not\in \dom(s)$ and for some $t\subseteq s$, $t \vdash x$. The intuition is that $x$ may be performed if the events in $s$ have occurred. 
Now, $\Hist(M)$, the set of histories over $M$, is defined inductively as the least family $H$ of consistent sets of events
which contains the empty set and is closed under accessibility, meaning that if $s \in H$ and $s \acc x$,
then for all $o \in O_x$,  $s \cup \{ (x,o)\} \in H$. Note that if a measurement can be performed, then any of its outcomes may occur, forming a valid history.

We can give a more explicit description of $\Hist(M)$ as a least fixed point. We define an increasing family of sets of histories $\{ H_k \}$ inductively:
\begin{align*}
    H_0 & \defeq  \{ \es \} \\
    H_{k+1} & \defeq  H_k \; \cup \; \{ s \cup \{ (x,o) \} \mid s \in H_{k}, s \acc x, o \in O_x   \}.
\end{align*}
If $X$ is finite, then for some $k$ we have $H_k = H_{k+1}$, and $\Hist(M) = H_k$ for the least such $k$.

\subsection{Strategies}

We regard a causal measurement scenario as specifying a game between Experimenter and Nature. Events $(x,o)$ correspond to the Experimenter choosing a measurement $x$, and Nature responding with outcome $o$. The histories correspond to the \emph{plays} or runs of the game.

Given this interpretation, we consider the notion of strategy. We focus first on the player Nature, whose strategies one may think of as hidden variables. This is in line with the usual discussion of contextuality, where the experimenter may freely choose which measurements to perform and such choices are beyond the scope of analysis.
Later, in \cref{sec:experimenterstrategies-adaptive}, we also consider the parallel notion of strategy for Experimenter, which can express adaptivity.

We define a \emph{strategy for Nature} over the game $M$ as a set of histories $\sg \subseteq \Hist(M)$ satisfying the following conditions:
\begin{itemize}
\item $\sg$ is downwards closed: if $s, t \in \Hist(M)$ and $s \subseteq t \in \sg$, then $s \in \sg$.
\item $\sg$ is deterministic and total: if $s \in \sg$ and $s \acc x$, then there is a unique $o \in O_x$ such that $s \cup \{ (x,o) \} \in \sg$.
\end{itemize}
Thus at any position $s$ reachable under the strategy $\sg$, the strategy determines a unique response to any measurement that can be chosen by the Experimenter.

We note an important property of strategies.
\begin{proposition}[Monotonicity]
If $s, t \in \sg$, $s \subseteq t$, and $s \acc x$, then 
\[ s \cup \{ (x,o) \} \in \sg \IMP t \cup \{ (x,o) \} \in \sg . \]
\end{proposition}
\begin{proof}
Under the given assumptions, since $t \acc x$, we must have $t \cup \{ (x,o') \} \in \sg$ for some $o' \in O_x$. Since $s \acc x$, we have 
that $s \cup \{ (x,o') \}$ is a history (\ie in $\Hist(M)$), 
and by down-closure, $s \cup \{ (x,o') \} \in \sg$. Since $\sg$ is deterministic, we must have $o = o'$.
\end{proof}

Monotonicity says that
the outcomes for a measurement $x$ under strategy $\sg$ are determined at the minimal histories at which $x$ can occur. This still leaves open the possibility of $\sg$ assigning different outcomes to $x$ relative to incomparable causal pasts.

We note another useful property, which follows immediately from totality and determinism.
\begin{proposition}[Maximality]
If $\sg$, $\tau$ are strategies with $\sg \subseteq \tau$, then $\sg = \tau$.
\end{proposition}

\subsection{The presheaf of strategies}

Given a causal measurement scenario $M = (X,O,{\vdash})$ and a set of measurements $U \subseteq X$, we define $M_U$, the restriction of $M$ to $U$, as the causal measurement scenario $(U, \{ O_x \}_{x \in U}, {\vdash_{U}})$, where $s \vdash_U x$ iff $s \vdash x$ and $\dom(s) \cup \{ x \} \subseteq U$.
Note that $M_X = M$.

\begin{proposition}
If $U \subseteq V$, then $\Hist(M_U)$ is a down-closed subset of $\Hist(M_V)$ under set inclusion.
\end{proposition}

Given a strategy $\sg$ over $M_V$, and $U \subseteq V$, we define $\sg |_U$, the restriction of $\sg$ to $U$, as the intersection $\sg |_U \defeq \sg \cap \Hist(M_U)$.
\begin{proposition}
If $\sg$ is a strategy over $M_V$ and $U \subseteq V$, then $\sg |_U$ is a strategy over $M_U$.
\end{proposition}
\begin{proof}
The restriction $\sg |_U$ inherits down-closure from $\sg$.
For the second condition, if $s \in \sg |_U$ and $s \acc_U x$, then $s \in \sg$ and $s \acc_V x$. So, there is a unique $o \in O_x$ such that $s \cup \{ (x,o) \} \in \sg$.
But since $x \in U$, we have $s \cup \{ (x,o) \} \in \Hist(M_U)$, and so $s \cup \{ (x,o) \} \in \sg |_U$.
\end{proof}

Given a causal measurement scenario $M = (X,O,{\vdash})$, we can now define a presheaf 
\[ \Gamma : \pow(X)^{\op} \to \Set \]
of strategies over $M$.
For each $U \subseteq X$, $\Gamma(U)$ is the set of  strategies for $M_U$. 
Given $U \subseteq V$, the restriction map $\Gamma(U \subseteq V) : \Gamma(V) \to \Gamma(U)$ is given by $\sg \mapsto  \sg |_U$. 

The following is immediate:
\begin{proposition}
$\Gamma$ is a presheaf.
\end{proposition}

\subsection{Historical note}

Causal measurement scenarios are a renaming and repurposing of Kahn--Plotkin information matrices \cite{kahn1993concrete}, which were introduced \textit{circa}~1975 to represent concrete domains.\footnote{For a historical perspective, see \cite{brookes1993historical}.}

We have changed the terminology to reflect the intuitions and applications motivating the present paper:

\vspace{.1in}
\begin{center}
\begin{tabular}{l|l}
Kahn--Plotkin & Here \\ \hline
information matrix & causal measurement scenario \\
cell & measurement \\
value & outcome \\
decision & event \\
configuration & history \\
\end{tabular}
\end{center}
\vspace{.1in}
The interpretation of causal measurement scenarios as Experimenter--Nature games, the notion of strategy, and the presheaf of strategies, are all new to the present paper.

\section{Causal contextuality}
Our plan now is to follow the script from \cite{abramsky2011sheaf}, replacing the event sheaf $\Event$ by the presheaf of strategies $\Gamma$.
Thus local sections are replaced by strategies, whose assignments of outcomes to measurements are sensitive to the previous history of the game.

A \emph{causal contexuality scenario} is a structure $(M, \Cover)$, where $M = (X, O, {\vdash})$ is a causal measurement scenario and $\Cover$ is a cover of $X$, \ie a family $\Cover = \{ C_i \}_{i \in I}$ of subsets of measurements $C_i \subseteq X$ satisfying $\bigcup\Cover = \bigcup_{i \in I} C_i = X$.
We work with the presheaf $\Gamma$ of strategies over $M$, as described in the previous section.

Recall the distribution monad $\Dist_R$ from \cite{abramsky2011sheaf}, where $R$ is a semiring.
When $R$ is the non-negative reals, it yields the usual discrete probability distributions.
We construct the presheaf $\Dist_R \Gamma$, obtained by composing the endofunctor part of the monad with the sheaf of strategies $\Gamma$. 

An \emph{empirical model} on the scenario $(M, \Cover)$ is a compatible family for the presheaf $\Dist_R \Gamma$ over the cover $\Cover = \{ C_i \}_{i \in I}$.
That is, it is a family $\{ e_i \}_{i \in I}$, where $e_i \in \Dist_R \Gamma (C_i)$,
subject to the compatibility conditions: for all $i, j \in I$, $e_i |_{C_i \cap C_j} = e_j |_{C_i \cap C_j}$.
Each distribution $e_i$ assigns probabilities to the strategies over $M_{C_i}$, \ie to those strategies over $M$ that only perform measurements drawn from the context $C_i$. As usual, the compatibility conditions require that the marginal distributions agree.
This follows the definition of empirical model in \cite{abramsky2011sheaf}, replacing the event sheaf by the presheaf of strategies.

The empirical model is \emph{causally non-contextual} if this compatible family extends to a global section of the presheaf $\Dist_R \Gamma$, \ie if there is a distribution $d \in \Dist_R \Gamma (X)$ such that, for all $i \in I$, $d |_{C_i} = e_i$.


If a causal contextuality scenario is finite, then so is the set of histories and therefore that of strategies. 
The causally non-contextual models thus form a convex polytope, the convex hull of the empirical models on $(M,\Cover)$ corresponding to deterministic strategies $\sg \in \Gamma(X)$.
This is in keeping with the usual setup of ``flat'' non-locality and contextuality (\ie without causality), where such classical polytopes are studied.
The classicality of a given model, \ie membership in this polytope, can be checked by linear programming;
and this also suggests a generalisation of the contextual fraction \cite{abramsky2017contextual} to the causal setting.

Similarly, causal contextuality is witnessed by violations of the linear inequalities defining the facets of the polytope.
An open question is to find a logical characterisation of such inequalities in the spirit of ``logical Bell inequalities'' \cite{abramsky2012logical}.

\section{Special cases}

To check that these notions make sense, we look at two special cases: flat scenarios and Gogioso--Pinzani scenarios.

\subsection{Flat scenarios}

A contextuality scenario from \cite{abramsky2011sheaf} is $(X, O, \Cover)$. We define the trivial enabling relation where all measurements are initially enabled: $\es \vdash x$ for all $x \in X$. This yields a causal measurement scenario $(M, \Cover)$, where $M = (X,O,{\vdash})$.

For any set of measurements $U \subseteq X$, the histories over $M_U$ have support contained in $U$.
Using the monotonicity property and the fact that all measurements are enabled by $\es$,
any strategy $\sg$ in $\Gamma(U)$ assigns the same outcome to each measurement across all its histories.
Hence, it will correspond to a section in $\Event(U) = \prod_{x \in U} O_x$. In fact, these will be in bijective correspondence.


Because of this bijective correspondence between $\Gamma$ and $\Event$, we see that the notions of empirical model, global section, and contextuality defined for the game-based scenario coincide with the usual notions in this case.

As this example illustrates, the restrictions on which measurements can be performed together are imposed by the cover, not by the causal structure. 

\subsection{GP scenarios}

In recent work, Stefano Gogioso and Nicola Pinzani studied a causal refinement of the sheaf-theoretic approach to non-locality over Bell scenarios \cite{gogioso2021sheaf}.

A GP scenario is given by $((\Omega, {\leq}), \{ \Iom \}_{\om \in \Om}, \{ \Oom \}_{\om \in \Om})$, where:
\begin{itemize}
\item $\Omega$ is a set of sites or agents (Alice, Bob, etc.), with a causal ordering.
\item $\Iom$ is the set of inputs (or measurement settings) at $\om$.
\item $\Oom$ is the set of outputs (or measurement outcomes) at $\om$.
\end{itemize}

Given such a scenario, we define a causal measurement scenario  $M = (X,O,{\vdash})$.
This mirrors the usual encoding of Bell non-locality scenarios as contextuality scenarios.
First, we set:
\begin{itemize}
\item $X \defeq \sum_{\om \in \Om} \prod_{\om' \leq \om}\Iomp = \{ (\om, \vi) \mid \om \in \Om, \vi = \{\vi(\om') \in \Iomp\}_{\om' \leq \om}\}$;
\item $O_{(\om,i)} \defeq \Oom$.
\end{itemize}
Given a set of events
\[
s = \{ \, (\,(\omega_1,\vi_1)\,,o_1),\, \ldots ,\, (\,(\om_n, \vi_n)\,, o_n) \,\}
\]
and a measurement $(\om, \vi) \in X$, we define
$s \vdash (\om, \vi)$ if and only if
the support of $s$ has a measurement for each site strictly preceding $\omega$, \ie
 $\{ \omega_1, \ldots, \omega_n \} = \{ \om' \in \Om \mid \om' < \om \}$, and moreover $\vi(\upsilon) = \vi_j(\upsilon)$ for all $\upsilon \leq \om_j$. The vector $\vi$ thus encodes all prior choices, as Nature's strategies are allowed to depend on them.
So, a measurement $(\om, \vi)$ can only be played after a measurement from each site in the causal past of $\om$ has been played.
Consequently, the support of any history is a set of measurements per site for some lower subset $\lambda \subseteq \Om$. 
This corresponds to the usual notion of context for Bell scenarios, refined to ensure that such contexts are ``causally secured''.

We consider a simple example to illustrate the comparison between $\Gamma$ defined over $(X, O, {\vdash})$, and the ``sheaf of sections'' from \cite{gogioso2021sheaf}.

We take $\Om$ to be the 2-chain $\om_1 < \om_2$. 
This is a variation on a standard bipartite Bell--CHSH type scenario, with Alice causally preceding Bob, and hence allowed to signal to Bob.
We take the standard Bell scenario cover, where the maximal contexts correspond to choosing one measurement per site, and focus our analysis on the contexts below the cover\footnotemark

\footnotetext{The equivalence between sections of $\Gamma$ and those of the presheaf from \cite{gogioso2021sheaf} actually extends more generally to all subsets of measurements, but this is sufficient to illustrate our main point.}

Now consider a strategy $\sg \in \Gamma(X)$.
The non-empty histories in $M$ which are compatible in the standard Bell cover  have the form
\[\{ \, (\,\om_1,\{ \omega_1 \mapsto z_1 \} \rangle)\,, o_1) \, \} \quad\text{ or }\quad \{ \, (\,(\om_1,\{ \omega_1 \mapsto z_1 \})\,, o_1),\, (\,(\om_2,\{ \omega_1 \mapsto z_1, \omega_2 \mapsto z_2\})\,, o_2) \, \} ,\]
where $z_i \in \{ x,y \}$, $o_i \in \{ 0,1\}$, $i=1,2$. 
Using monotonicity, the strategy $\sg$ assigns a unique $o_1$ for each $(\om_1,z_1)$ and a unique $o_2$ for each $(\om_1,z_1)$ and $(\om_2,z_2)$.
Thus $\sg$ determines a pair of functions of type
\[(I_{\om_1} \to O_{\om_1}) \quad\times\quad (I_{\om_1} \times I_{\om_2} \to O_{\om_2}). \]
This accords with the description given in \cite{gogioso2021sheaf}; see in particular the discussion in Section~5.
It extends to an equivalence between $\Gamma$ and the sheaf of sections of \cite{gogioso2021sheaf}.

Thus, if we take the standard Bell cover we obtain the same empirical models and notion of contextuality as in \cite{gogioso2021sheaf}.
In an extended version of the present paper, we show that
this analysis carries over to general GP scenarios. Hence, we recover the Gogioso--Pinzani theory as a special case of our framework.

\section{The sheaf property for the strategy presheaf}

The strategy presheaf $\Gamma$ plays the role in our causal theory of the event sheaf $\Event$ in \cite{abramsky2011sheaf}.
The sheaf property of $\Event$ has some conceptual significance since it shows that for deterministic models local consistency implies global consistency. It is only when we introduce distributions, whether probabilistic or possibilistic, that the sheaf property fails and contextuality arises.
This raises the question of whether $\Gamma$ is also a sheaf.

Let $\{ U_i \}_{i \in I}$ be a family of subsets of $X$ covering $U = \bigcup_{i \in I} U_i$.
Suppose we are given a compatible family $\{ \sigma_i \}_{i \in I}$, with $\sg_i \in \Gamma(U_i)$
and $\sg_i |_{U_i \cap U_j} = \sg_j |_{U_i \cap U_j}$ for all $i, j \in I$.
The sheaf property requires that there exist a unique strategy $\sg \in \Gamma(U)$ such that $\sg |_{U_i} = \sg_i$ for all $i \in I$.

From the definition of restriction, if such a gluing $\sg$ exists, it must contain the union $\sg' \defeq \bigcup_{i \in I} \sg_i$.
So, if this $\sg'$ happens to be a strategy, by maximality it must be the required unique gluing of the family $\{\sg_i\}_{i\in I}$.
The union of down-closed sets is down-closed. 
Thus $\sg'$ can only fail to be a strategy if determinacy or totality fails.
We show that the first of these can never arise.

\begin{proposition}
If $\{\sigma_i\}_{i \in I}$ is a compatible family for the presheaf $\Gamma$, then $\sg' \defeq \bigcup_{i \in I} \sg_i$ is deterministic.
\end{proposition}
\begin{proof}
Suppose that $s \cup \{ (x,o_k) \} \in \sigma'$ for $k = 1,2$. 
For some $i,j \in I$ we have $s \cup \{ (x,o_1) \} \in \sigma_i$ and $s \cup \{ (x,o_2) \} \in \sigma_j$.
This implies that $\dom(s) \cup \{x\} \subseteq U_i \cap U_j$, and hence $s \cup \{ (x,o_1) \} \in \sigma_i |_{U_i \cap U_j}$ and $s \cup \{ (x,o_2) \}  \in \sigma_j |_{U_i \cap U_j}$. By compatibility and determinacy of $\sg_i$ and $\sg_j$, this implies $o_1 = o_2$.
\end{proof}

In general, it may not be possible to complete $\sg'$ to a total strategy, and if such an extension does exist, it may not be unique.
We give simple examples to show how these can happen.

\begin{example}
Fix $X = \{ x,y,z \}$, $O_w = \{ 0,1\}$ for all $w \in \{ x,y,z\}$, and the following enabling relation:
\[ \es \vdash x, \qquad \es \vdash y, \qquad \{  (x,0) \} \vdash z, \qquad \{  (y,0) \} \vdash z .\]
Consider the cover consisting of $U_1 \defeq \{x,z\}$ and $U_2 \defeq \{ y,z \}$, and take strategies
\[\sg_1 \defeq \{ \, \es, \, \{ (x,0) \}, \, \{ (x,0), (z,0) \} \, \} \quad\text{ and }\quad \sg_2 \defeq \{ \, \es, \, \{ (y,0) \}, \, \{ (y,0), (z,1) \} \, \} \, . \]
Note that $\sg_1$ and $\sg_2$ are compatible since they both restrict to the empty strategy over $U_1 \cap U_2 = \{ z \}$, as the measurement $z$ is not enabled. Similarly, $\sg_1$ and $\sg_2$ are both total. However, $\sg_1 \cup \sg_2$ is not total, nor can it be completed to a total strategy. Note that $y$ is accessible from $s= \{ (x,0), (z,0) \}$ so an extension of $s$ must assign an outcome to $y$, which must be equal to zero due to downward-closedness. Following the same reasoning for $x$, we are forced to include both  $\{ (x,0), (y,0), (z,0) \}$ and $\{ (x,0), (y,0), (z,1) \}$ in the strategy, contradicting determinism.
\end{example}

\begin{example}
Fix $X = \{ x,y,z \}$, $O_w = \{ 0,1\}$ for all $w \in \{ x,y,z\}$, and the following enabling relation:
\[ \es \vdash x, \qquad \es \vdash y, \qquad \{  (x,0),  (y,0) \} \vdash z .\]
Consider the cover consisting of $U_1 \defeq \{ x,z\}$ and $U_2 \defeq \{ y,z \}$, and take strategies
\[\sg_1 \defeq \{ \, \es, \, \{ (x,0) \} \, \} \quad\text{ and }\quad \sg_2 \defeq \{ \, \es, \, \{ (y,0) \} \, \} \, . \]
Note that $\sg_1$ and $\sg_2$ are compatible since they both restrict to the empty strategy over $U_1 \cap U_2 = \{ z \}$, as the measurement $z$ is not enabled.
Similarly, $\sg_1$ and $\sg_2$ are both total, since $z$ is not accessible from any history over $U_1$ or $U_2$. However, $\sg_1 \cup \sg_2$ is not total,
since $y$ is accessible from $\{ (x,0) \}$ and $x$ is accessible from $\{ (y,0) \}$.
There is a unique choice of outcomes that can be assigned to these variables leading to restrictions to $\sg_1$ and $\sg_2$ as required for a gluing. Both lead to the history $\{ (x,0), (y,0) \}$.
However, $z$ is now accessible from this history, and there are no constraints on the value assigned to it, so we lose uniqueness.
\end{example}

This example is rather pathological, as it hinges on the inaccessibility of $z$ in the cover, leading to the following question.

\begin{question}
Is there a notion of ``good cover'' which implies that gluings exist and are unique?
\end{question}

\textbf{Note added in proof} We have found a positive answer to this question. If we require that the cover comprises sets of measurements which are causally secured in an appropriate sense, then the sheaf property holds. This will be described in detail in a sequel to the present paper.

\section{Experimenter strategies and adaptive computation}
\label{sec:experimenterstrategies-adaptive}

The strategies considered so far have been strategies for Nature. These prescribe a response -- an outcome -- for each measurement that can be chosen by the Experimenter.
Using the duality inherent in game theory, there is also a notion of strategy for Experimenter.
To formulate this, we use the following observation.

\begin{proposition}
For a history $s \in \Hist(M)$, the following are equivalent:
\begin{enumerate}
\item $s$ is maximal in $(\Hist(M),{\subseteq})$;
\item no measurement is accessible from $s$, \ie for all $x \in X$, $\neg(s \acc x)$.
\end{enumerate}
\end{proposition}

We now define a \emph{strategy for Experimenter} over the game $M$ to be a set of histories $\tau \subseteq \Hist(M)$ satisfying the following conditions:
\begin{itemize}
\item $\tau$ is downwards closed: if $s, t \in \Hist(M)$ and $s \subseteq t \in \tau$, then $s \in \tau$.
\item $\tau$ is co-total: if $s \in \tau$
and $s$ is not maximal,
then there is a measurement $x$
with $s \acc x$ such that $s \cup \{ (x,o) \} \in \tau$ for some $o \in O_x$.
Moreover, for all such $x$, $s \cup \{ (x,o') \} \in \tau$ for all
$o' \in O_x$.
\end{itemize}
Thus at each stage, the strategy determines which measurements may be performed.
Note that it may allow more than one measurement, so some nondeterminism remains.

For each such measurement, it must then accept any possible response from Nature. The future choices of the Experimenter can then depend on Nature's responses, allowing for adaptive protocols.

If we are given a strategy for Nature $\sg$ and a strategy for the Experimenter $\tau$, we can play them off against each other, resulting in $\langle \sigma \mid \tau \rangle \defeq \sg \cap \tau$. 
This is the down-set of a set of maximal histories.
This operation can be extended to distributions on strategies, \ie to \emph{mixed strategies}, in a bilinear fashion.\footnote{The extension to mixed strategies hinges on the fact that the distribution monad is commutative.}

We refer to strategies for Nature as N-strategies, and to strategies for Experimenter as E-strategies. 


\subsection{Anders--Browne revisited}

We now show how the Anders--Browne construction of an \textsc{OR} gate discussed in section~\ref{ABsec} can be formalised using an Experimenter strategy.

First, we have the description of the standard GHZ construction. This is given by a flat measurement scenario with $X = \{ A_i, B_j, C_k \mid i,j,k \in \{ 0,1\} \}$, and $O_x = \{ 0,1 \}$ for all $x \in X$.
The maximal compatible sets of measurements are all sets of the form $\{ A_i, B_j, C_k \}$ with $i,j,k \in \{ 0,1\}$, \ie a choice of one measurement per each site or agent.
We regard each measurement as initially enabled. The N-strategies for this scenario form the usual sections assigning an outcome to each choice of measurement for each site, and the GHZ model assigns distributions on these strategies as in the table shown in section~\ref{ABsec}.

To get the Anders--Browne construction, we consider the E-strategy which initially allows any $A$ or $B$ measurement to be performed, and after a history $\{ (A_i, o_1), (B_j, o_2) \}$ chooses the $C$-measurement $C_{i \oplus j}$.
Playing this against the GHZ model results in a strategy that computes the \textsc{OR} function with probability 1.

The full power of adaptivity is required when using this as a building block to implement a more  involved logical circuit. Suppose that the output of the \textsc{OR} gate above is to be fed as the first input of a second \textsc{OR} gate, built over a GHZ scenario with measurements labelled $\{ A'_i, B'_j, C'_k \mid i,j,k \in \{ 0,1\} \}$.
The $E$-strategy implements the first \textsc{OR} gate as above, with any $B'$ measurement also enabled, being a free input.
After that, the $A'$-measurement can be determined: after a history containing $\{ (A_i, o_1), (B_j, o_2), (C_{i \oplus j}, o_3) \}$, the E-strategy chooses the $A'$-measurement $A'_{o_1 \oplus o_2 \oplus o_3}$. The second \textsc{OR} gate is then implemented like the first. Note that the choice of $A'$-measurement depends not only on previous measurement choices, but on outcomes provided by Nature.

\section{Outlook}

In a forthcoming extended version of this paper, we show how a number of additional examples, including Leggett--Garg, can be handled in our approach.
We also show that our formalism faithfully represents a number of others, including Gogioso--Pinzani scenarios, adaptive MBQC, and causal networks.
In forthcoming related work 
we incorporate a form of memory (or look-back) restriction in some simple scenarios whereby Nature may only remember the $k$ most recent events, and obtain a \Vorobev-type theorem \cite{vorobev1962consistent,barbosa2015DPhil} in that setup.

In future work, we aim to employ our formalism to describe unconditional quantum advantage in shallow circuits, building on \cite{bravyi2018quantum,aasnaess2022comparing}.
We will also investigate other applications to quantum advantage.

We also aim to clarify how our approach can be related to the currently very active study of indefinite causal orders \cite{oreshkov2012quantum,chiribella2013quantum}.

The game formulation opens up the possibility of tapping into the sophisticated literature in game theory, \eg on properties of game trees \cite{kuhn1953extensive,isbell1958finitary}. This is likely to offer important pointers for further development of our framework.
For example, the concepts of randomised strategies (and the relationship between global and local randomisations) or memory properties are likely to be fruitful.
A similar source of potential inspiration is the  literature on game semantics, where game concepts are used to model a wide array of programming language features; see e.g.~\cite{murawski2016invitation} for a recent overview.

\subsection*{Acknowledgements}
This work was developed in part while AS was hosted on secondment at INL.

This work is supported by the Digital Horizon Europe project FoQaCiA, \textit{Foundations of quantum computational advantage}, GA no.~101070558, funded by the European Union, NSERC (Canada), and UKRI (U.K.).

SA also acknowledges support from EPSRC -- Engineering and Physical Sciences Research Council (U.K.) through
EPSRC fellowship EP/V040944/1, \textit{Resources in Computation}. 
RSB also acknowledges support from FCT -- Funda\c{c}\~{a}o para a Ci\^encia e a Tecnologia (Portugal) through CEECINST/00062/2018.
AS acknowledges support from EPSRC Standard Research Studentship (Doctoral Training Partnership), EP/T517811/1, and the Smith-Westlake Graduate Scholarship at St. Hugh's College.

We thank the journal referees for their comments, which suggested several improvements in the presentation.

\bibliographystyle{amsplain}
\bibliography{ccc}

\vspace{2cm}

\end{document}